\documentclass[runningheads]{llncs}

\newif\iffullversion
\fullversiontrue

\usepackage{cite}
\usepackage{amsmath,amssymb,amsfonts}
\usepackage{algorithmic}
\usepackage{graphicx}
\usepackage{textcomp}
\usepackage{xcolor}

\usepackage[T1]{fontenc}
\usepackage[utf8]{inputenc}
\usepackage{pdfpages}
\usepackage{pgfplotstable}
\usepackage{pgfplots}
\pgfplotsset{compat=1.15}
\usepackage{caption}

\usepackage{filecontents}

\usepackage{flushend}

\usepackage{enumitem}
\setlist{noitemsep,topsep=0.03in} 
\usepackage[T1]{fontenc}
\usepackage[utf8]{inputenc}
\usepackage{cleveref}

\usepackage{booktabs}
\usepackage{graphicx}
\usepackage{marginnote}
\usepackage{microtype}
\usepackage{multicol}
\usepackage{multirow}
\usepackage[normalem]{ulem}
\usepackage{wrapfig}
\usepackage{xspace}

\usepackage{algorithm2e}
\usepackage{amsfonts}
\usepackage{amsmath}
\usepackage{amssymb}
\usepackage{centernot}

\newcommand{\DBKE}{\styleAlgorithm{DBKE}}
\newcommand{\Sym}{\styleAlgorithm{Sym}}
\newcommand{\Wrap}{\styleAlgorithm{Wrap}}
\newcommand{\Unwrap}{\styleAlgorithm{Unwrap}}
\newcommand{\Degrade}{\styleAlgorithm{Degrade}}

\usepackage{tikz}
\usetikzlibrary{arrows,backgrounds,positioning,fit,shadows,shapes.geometric,decorations.pathreplacing}

\newcommand{\lightparagraph}[1]{\smallskip\textit{#1}}

\newcommand{\getsr}{{\:\leftarrow\hspace*{-3pt}\raisebox{.5pt}{$\scriptscriptstyle\$$}\:}}
\newcommand{\tor}{{\:\raisebox{.5pt}{$\scriptscriptstyle\$$}\hspace*{-3pt}\rightarrow\:}}
\newcommand{\bit}{\{0,1\}}

\newcommand{\true}{\texttt{true}}

\newcommand{\Adversary}{\mathcal{A}}
\newcommand{\Bdversary}{\mathcal{B}}

\newcommand{\Adv}[2]{\mathrm{Adv}^{\styleSecurityNotion{#1}}_{#2}}
\newcommand{\Exp}[2]{\mathrm{Exp}^{\styleSecurityNotion{#1}}_{#2}}

\newcommand{\styleSecurityNotion}[1]{\textsf{\upshape #1}}

\newcommand{\styleAlgorithm}[1]{\textrm{\upshape #1}}
\newcommand{\Dec}{\styleAlgorithm{Dec}}

\newcommand{\Enc}{\styleAlgorithm{Enc}}

\newcommand{\gamechange}[2][blue]{\begin{tikzpicture}[baseline=(X.base)] \node[rectangle,draw,#1] (X) at (0,0) {\color{#1} #2}; \end{tikzpicture}}

\newcommand{\seed}{r}
\newcommand{\partialseed}{\overline{r}}
\newcommand{\salt}{s}
\newcommand{\checksum}{h}
\newcommand{\key}{k}
\newcommand{\difficulty}{d}
\newcommand{\maxdifficulty}{D}
\newcommand{\wrappedkey}{w}
\newcommand{\ctxt}{c}
\newcommand{\msg}{m}
\newcommand{\sklen}{\lambda}

\newcommand{\KeySp}{\mathcal{K}}
\newcommand{\MsgSp}{\mathcal{M}}

\newcommand{\bbN}{\mathbb{N}}

\newcommand{\shortlongeqn}[2][.]{$#2$#1}

\begin{document}
\title{ArchiveSafe: Mass-Leakage-Resistant Storage from Proof-of-Work}
\author{Moe Sabry\inst{1} \and
Reza Samavi\inst{2} \and
Douglas Stebila\inst{3}}
\institute{McMaster University, Canada \email{alym2@mcmaster.ca} \and
Ryerson University, Vector Institute, Canada \email{samavi@ryerson.ca} \and
University of Waterloo, Canada \email{dstebila@uwaterloo.ca}
}
\maketitle              \begin{abstract}
Data breaches---mass leakage of stored information---are a major security concern. Encryption can provide confidentiality, but encryption depends on a key which, if compromised, allows the attacker to decrypt everything, effectively instantly. Security of encrypted data thus becomes a question of protecting the encryption keys.
In this paper, we propose using \emph{keyless encryption} to construct a \emph{mass leakage resistant archiving system}, where decryption of a file is only possible after the requester, whether an authorized user or an adversary, completes a \emph{proof of work} in the form of solving a cryptographic puzzle. 
This proposal is geared towards protection of infrequently-accessed \emph {archival data}, where any one file may not require too much work to decrypt, decryption of a large number of files---mass leakage---becomes increasingly expensive for an attacker.
We present a prototype implementation realized as a user-space file system driver for Linux.  We report experimental results of system behaviour under different file sizes and puzzle difficulty levels.
Our keyless encryption technique can be added as a layer \emph{on top of} traditional encryption: together they provide strong security against adversaries without the key and resistance against mass decryption by an attacker.

\iffullversion 
\fi

\end{abstract}

\section{Introduction}\label{section:Introduction}

Attacks on information systems have become increasingly common.
Whatever the attack vector, a frequent outcome is a \emph{data breach}, in which a large volume of sensitive information is stolen from the victim organization.
Archival data---stored indefinitely but not regularly accessed---has been targeted in many data breaches~\cite{ameritrade, adobe, bankny}, leading to loss of privacy, loss of reputation, business setbacks, and costly remediation.

Modern IT security protection techniques focus on \emph{defense-in-depth}, one component of which is encryption of data at rest to support confidentiality.  However, encryption, even when implemented using secure, carefully implemented algorithms, is typically all-or-nothing: if the key is secure, the attacker learn virtually nothing, and the attack cannot succeed, but once the key is compromised, the attacker can decrypt everything, with minimal overhead.

Hardware-assisted cryptography, such as hardware security modules (HSMs), trusted computing, or secure enclaves like Intel SGX\footnote{https://software.intel.com/en-us/sgx} or ARM TrustZone\footnote{https://developer.arm.com/ip-products/security-ip/trustzone} may prevent keys from leaking if decryption is only ever done inside a trusted module, but many IT systems remain software-only without use of these technologies.

\paragraph{Scenario and goals.}
Against these types of threats, we aim develop a \emph{mass leakage resistant archiving system} with the goal of enhancing defense-in-depth for encryption.  
We aim to preserve confidentiality even in the presence of an adversary with full access to the system, including ciphertexts and decryption keys.
While no system can provide full cryptographic security in the face of such a well-informed adversary, our goal is to increase the economic cost of \emph{mass leakage}, which for our purposes is defined as an adversary obtaining the plaintexts of a large number of files or database records, not just one.

Unlike most applications of cryptography, we do not aim to achieve a difference in work factor between honest parties and adversaries. Rather, we assume that honest parties and adversaries have different \emph{goals}, and we aim to change the economics of data breaches by achieving a difference in the cost of honest parties and adversaries achieving their goals.
In our scenario, honest parties need to store a large number of files, but only access a small number of them.  Consider for example a tax agency: after processing millions of citizens' tax returns each year, those files must be stored for several years in case an audit or further analysis is required, but only a small fraction of those records will end up actually being pulled for analysis.  In contrast, an adversary breaching the tax agency's records may want to read a large number of files to identify good candidates for identity theft or other criminal actions.

\subsection{Contributions}

We design a system, called \emph{ArchiveSafe}, where access to a resource is only possible after the requester---whether an honest user on adversary---has expended sufficient computational effort, in the form of solving a ``moderately hard'' proof-of-work or cryptographic puzzle \cite{dwork1992pricing}.  Since we will not rely on the access control system nor any keys to be uncompromised, the decryption operation itself must be tied to the cryptographic puzzle.  
In our approach, while a proper cryptographic key is used to encrypt a file, the encryption key is not stored, even for legitimate users.
Instead, the key is wrapped in a proof-of-work-based encryption scheme with a desired difficulty level, and all users---adversarial or honest---must perform the proof-of-work to recover the key and then decrypt the file. 

Our main technical tool for building of ArchiveSafe is a new cryptographic primitive that we call \emph{difficulty-based keyless encryption} (DBKE), which is an encryption scheme that does not make use of a stored key.
We give a generic construction for DBKE from a standard symmetric encryption scheme and a new tool called \emph{difficulty-based keyless key wrap}, which wraps the symmetric encryption key in an encapsulation that can only be unwrapped by performing a sufficiently high number of operations, as in a proof-of-work scheme.
Difficulty-based keyless key wrap can be achieved from many types of cryptographic puzzles, and we show one example based on hash function partial pre-image finding \cite{juels1999client,JJ99}.
One interesting feature of using this form of hash-based puzzle, which to our knowledge is a novel observation on hash-based puzzles, is that the puzzle and ciphertext can be \emph{degraded}---i.e., turned into a harder one---essentially for free.
We use the reductionist security methodology to formalize the syntax and security properties of difficulty-based keyless encryption and keyless key wrap and show that our hash-based construction achieves these properties.

\begin{figure}[t]
\begin{center}
\scalebox{0.8}{
\begin{tikzpicture}[xscale=1.25,yscale=-0.4]
\node[rectangle,draw,minimum height=0.25in] at (0,-0.7) {Application};
\node[rectangle,draw,minimum height=0.25in] at (3,-0.7) {ArchiveSafe Driver};
\node[rectangle,draw,minimum height=0.25in] at (6,-0.7) {Underlying Storage};
\draw[->] (0,1.25)-- node[above] {write $\msg$} +(2,0);
\node at (3,2) {generate puzzle, key};
\node at (3,3) {$\ctxt \gets \Enc(\key, \msg)$};
\draw[->] (4,4.25)-- node[above] {write $puz$, $\ctxt$} +(2,0);
\draw[dotted,thick] (-0.5,5) -- +(7.5,0);
\draw[->] (0,6.25) -- node[above] {read} +(2,0);
\draw[->] (4,7.25) -- node[above] {read} +(2,0);
\draw[<-] (4,8.25) -- node[above] {puz, $\ctxt$} +(2,0);
\node at (3,9) {$\key \gets \mbox{Solve}(puz)$};
\node at (3,10) {$\msg \gets \Dec(\key, \ctxt)$};
\draw[<-] (0,11.25) -- node[above] {$\msg$} +(2,0);
\end{tikzpicture}
}
\end{center}
\vspace{-1em}
\caption{High-level overview of ArchiveSafe, showing a write followed by a read.}
\label{fig:high-level}
\vspace{-1em}
\end{figure}
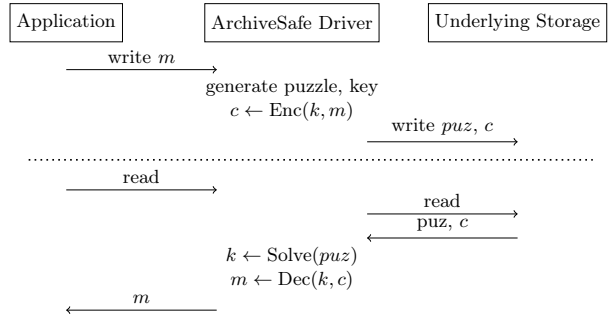

\Cref{fig:high-level} gives a high-level overview of how an application interacts with the ArchiveSafe system.  The two main operations performed by the ArchiveSafe system are (i) creating a puzzle and encrypting during writes, and (ii) solving the puzzle and decrypting during reads.
ArchiveSafe could be used in a variety of data storage architectures: on a local computer; on a file server; or in a cloud architecture.  
In a file server or cloud scenario, an IT system may be set up so the file server enforces that all files are protected by ArchiveSafe during writes by centralizing puzzle creation and encryption, but leaves puzzle solving and decryption to clients.  Since puzzle creation and encryption in our system is cheap, this avoids bottlenecks on the file server.  Individual client applications occasionally reading a small number of files have to do a moderate, but not prohibitive, amount of work to solve the puzzle to obtain the key to decrypt.

We build a prototype implementation showing the use of ArchiveSafe on a local computer.
Our prototype is implemented as a filesystem-in-userspace (FUSE) driver on Linux.
A FUSE driver can be used to intercept I/O operations in certain directories (mount points) before reading/writing to disk.
This allows us to implement ArchiveSafe in a manner that is transparent to the application, as well as transparent to the underlying storage mechanism, which could be a local disk (with normal disk encryption enabled or not), or a network share mounted locally.
We validate the performance of our prototype implementation, focusing primarily on ensuring that write operations incur minimal overhead.
(Since system administrators can set policies with puzzle difficulties requiring seconds or minutes of computational effort to solve, slow read performance is \emph{intended}, and there is little sense in performance measurements on reads, beyond checking that they scale as intended with no unexpected overhead.)
We envision that, when used on a local computer, ArchiveSafe would be applied only to a subset of the directories on the computer.  One might use ArchiveSafe to protect documents created by the user more than a certain number of days ago, but would not use it on system libraries and executables.

We highlight that ArchiveSafe is meant to add \emph{defense-in-depth} to confidentiality: one would typically not rely on ArchiveSafe alone, but combine it with traditional encrypted file system or database encryption.  
In this combination, traditional encryption using strong algorithms and keys, provides a high level of security if the keys are not compromised, but we still have the difficulty-based keyless encryption of ArchiveSafe as a bulwark if the keys are compromised. To succeed under this setup, the adversary must compromise the traditional encryption keys in addition to solving a large number of DBKE puzzles corresponding to the files in the archive.

\subsection{Related Work}

\paragraph{Filesystem encryption.}
Blaze \cite{blaze1993cryptographic} introduced the Cryptographic File System (CFS). CFS uses a different key for each directory, and the user is required to enter the key in every session to access the directory and its contents.  
Subsequent proposals include the Transparent Cryptographic File System (TCFS) \cite{cattaneo2001design}, Cryptfs \cite{zadok1998cryptfs} and Ncryptfs \cite{wright2003ncryptfs}. 
In recent years, encrypted filesystems have become widespread, and all major operating systems provide implementations, often enabled by default (FileVault on Apple's macOS\footnote{https://support.apple.com/en-ca/HT204837}, BitLocker on Microsoft Windows\footnote{https://docs.microsoft.com/en-us/windows/security/information-protection/bitlocker/bitlocker-overview}, and a range of options on Linux such as Linux Unified Key Setup (LUKS)\footnote{https://guardianproject.info/archive/luks/}).
The common practice in these technologies is to use a single master  key from which multiple keys are derived per-file, per-directory, or per-sector; the master key is usually stored on the device itself, encrypted under the user's password.  Once the user has logged in, the filesystem transparently and automatically decrypts files.

Over the past decade, there has been much research on encrypted databases (e.g., \cite{cryptdb,arx,dbpuzz}) that retain some functionality for legitimate users, for example using order-preserving encryption so that sorting a column of ciphertexts yields approximately the same order as if the plaintexts were sorted.  This increased functionality comes at the cost of information leakage, and there is an extensive debate in the literature about these techniques.

\noindent
\paragraph{Proof-of-work systems.}
Dwork and Naor \cite{dwork1992pricing} introduced client puzzles to control junk email: recipients would only accept emails if the sender was able to solve a puzzle.  It should be ``moderately hard'' for the sender to solve the puzzle, but easy for recipient to check whether a solution is valid.  This was the first example of a proof-of-work system, which in general grants access to a resource dependent on the requester being able to demonstrate proof that they have performed some work, typically in the form of solving a puzzle.  Client puzzles were for many years suggested as a means to prevent denial of service attacks in a range of contexts \cite{juels1999client,aura2000resistant,dean2001using,waters2004new,suriadi2011defending,rangasamy2012effort}, but have seen renewed interest as a building block for cryptocurrencies and blockchains. Client puzzles are generally classified either based on their limiting factors in solving the puzzle (CPU-bound versus memory-bound) or based on whether the operations required to solve the puzzle is parallelizable.
The simplest CPU-bound puzzles are based on cryptographic hash functions, such as: finding a preimage of a hash given a hint (e.g., a part of the preimage) \cite{juels1999client,JJ99}; or finding an input whose hash starts with a certain number of zero bits \cite{Bac04}. 
Non-parallelizable CPU-bound puzzles often rely on a number of theoretical approaches.  For example, \cite{rivest1996time} uses repeated squaring modulo an RSA modulus. 
Memory-bound puzzles \cite{abadi2005moderately,dwork2003memory} use techniques for which the best known solving algorithm involves a large number of memory accesses; it is argued that memory access time varies less than CPU speed between small and large computing platforms, and that building customized hardware is more expensive for memory-bound puzzles.

\noindent\paragraph{Proof-of-work systems for confidentiality.}
In \cite{rivest1996time}, time-lock encryption was proposed as a way of ``sending information into the future'', and focused specifically on hiding keys or data in a proof-of-work system that had a predictable wall-clock time for solving, thus focusing on puzzles for which the best known solving algorithm is inherently sequential.  Vargas et al. \cite{vargas2018mitigating} designed a database encryption system called ``Dragchute'' based on time-lock encryption, aiming to provide both confidentiality and the ability to demonstrate compliance with retention laws.  
Each ciphertext in this system is accompanied by an authentication tag which contains a non-interactive zero-knowledge proof.
 Solving the puzzle will yield a valid decryption key for the ciphertext; moreover, the proof can be checked much more efficiently than the full work required to solve and decrypt the ciphertext.  A simpler database encryption scheme relying on hash-based client puzzles, without any efficient verification of well-formedness, was proposed by Moghimifar \cite{dbpuzz}.

\section{Requirements}\label{section:Requirements}

In this section, we discuss the functionality and security requirements for a mass leakage resistant archiving system, which informs our construction and evaluation in subsequent sections.

\subsection{Design Criteria}\label{section:DesignCriteria}

\lightparagraph{Confidentiality in the face of compromised keys.}
The system should achieve some level of confidentiality even if all stored keys are compromised.
This means we assume that an adversary can learn a symmetric key or a private key corresponding to a public key
 stored for later use in decrypting a ciphertext, even if the key is stored in a separate key management service, trusted computing or secure enclave environment, or separate tamper-resistant device. 

\lightparagraph{Cooperation with traditional encryption.}
It should be possible to use the system in conjunction with the traditional encryption mechanisms applied to storage systems (folder/disk encryption, database encryption, etc.), so that strong confidentiality is achieved if keys are not compromised, but some confidentiality is retained in the face of compromised keys.

\lightparagraph{Reliance on industry standard cryptographic algorithms.}
Deployed IT systems should rely only on well-vetted, standardized cryptographic algorithms.
But all such algorithms for achieving confidentiality---public key or symmetric---require a secret key, seemingly conflicting with the first design criteria of confidentiality in the face of compromised keys.  
Our construction builds a mechanism for confidentiality without keys while still relying on standard cryptographic algorithms like AES for symmetric encryption: while a proper cryptographic key is used to encrypt data, that key is not kept, even by authorized users.
Instead, the key is wrapped in a proof-of-work-based encryption scheme with a desired difficulty level, and users must solve the proof-of-work to recover the key and then decrypt the data.  
We introduce difficulty-based keyless encryption in \Cref{section:DBKE} which formalizes this idea and generically construct it from standard cryptographic algorithms such as AES and Argon2.

\lightparagraph{Imposing a significant cost to access a large number of files while maintaining acceptable cost to access one file.}
Since we do not have a key that gives honest users an advantage over the adversary, we should look at things from the viewpoint of typical honest behaviour---periodically accessing a small number of files---versus adversary behaviour---accessing a large number of files in a data breach.
Proof-of-work and related techniques have long been used to achieve security goals from that viewpoint, whether in password hardening or client puzzles for denial of service resistance.  

\lightparagraph{Customizing file access cost.}
It should be possible for a system administrator or user to control the cost incurred by the adversary or honest user for accessing a file.
This may be set as a system-wide policy or a file-by-file basis, depending on the desired access control paradigm.
This is achieved in our system by varying the difficulty level of the puzzle wrapping the decryption key.

A related design criteria is the ability to customize file access cost \emph{over time}.
Demand for access to records may change over time; for example, records older than 5 years may be accessed much less frequently than more recent records.
Our system allows the file access cost to be increased with minimal effort, through a process we call \emph{puzzle degradation}, that could be performed as part of regular system maintenance.
This is a novel feature available from some types of puzzle constructions but not others, and in particular not from the number-theoretic repeated squaring non-parallelizable constructions used in time-lock puzzles \cite{rivest1996time} and the Dragchute database encryption system \cite{vargas2018mitigating}.

\subsection{Choice of Puzzle}
One of the major design decisions for our system is which type of puzzles to use: sequential versus parallelizable, and CPU-bound versus memory-bound.  

As our design criteria focus on mass leakage adversaries trying to decrypt \emph{many} files, and since we think of cost in a general economic sense, we do not have to restrict to proof-of-work mechanisms that are sequential/non-parallelizable.  Concerned with an adversary trying to decrypt many files who has parallel computing resources available to them, it does not matter whether they choose to deploy their parallel resources to sequentially decrypt each file quickly or in parallel decrypt many files more slowly. Overall, they will decrypt the same number of files with the same resources.  We also need not worry about the variability of puzzle solving time for individual instances, only the expected puzzle solving time for many instances.  These design choices are, for example, significantly different from those of the Dragchute system for database confidentiality and integrity from proof-of-work.  Moreover, parallelization permits honest users to reduce the latency in occasional access of files by taking advantage of short, on-demand use of cloud servers (see \Cref{tab:cost}).

Whereas sequential versus parallelizable puzzles is a qualitative choice for our scenario, CPU-bound versus memory-bound is a quantitative choice with respect to the economic cost.  To achieve a given dollar-cost-for-adversary, it is possible to pick appropriate parameters for both CPU-bound and memory-bound puzzles under appropriate cost and puzzle-solving assumptions.  So, a priori, either can be used in our constructions.  For our prototype we choose simple hash-based CPU-bound puzzles because puzzle creation is cheaper (thereby achieving extremely low overhead on write operations) and because they allow us to obtain novel useful functionality such as puzzle degradation (\Cref{section:PuzzleDegradation}), but with the hash function being Argon2 which is designed to be resistant to GPU and ASIC optimization.  
Picking appropriate difficulty levels for puzzles is something an adopter must do as a function of the tolerable cost for honest users to access data, the perceived risk of a data breach, and the anticipated value of the information to an adversary.  We do not aim to study such economic calculations exhaustively, but we provide one worked example in \Cref{section:Discussion} and \Cref{tab:cost}.

\subsection{Threat Model}

ArchiveSafe is a software system with one target asset, the data files. The security goal for the target asset is confidentiality.  
As shown in \Cref{fig:high-level}, information flows from the user application through the ArchiveSafe driver to the underlying storage during writes, and in the reverse direction during reads.  

An adversary could access the system either via the same mechanism as an honest user application (i.e., mediated by the ArchiveSafe driver), or may have direct access to the underlying storage.
We aim to achieve confidentiality against a strong adversary that can bypass the ArchiveSafe driver during read operations (e.g., because they are untrusted server administrators, or because they have compromised the kernel using privilege escalation), or who can directly read from the underlying storage (e.g., an untrusted cloud storage provider, or physical theft of a hard drive).
We do not consider in our threat model an adversary who undermines the write operation to intercept data during a write operation or who prevents the ArchiveSafe technique from being applied when saving files. We assume operations by honest parties are performed on a trusted and uncompromised system that faithfully deletes keys from memory once an operation is completed.

\section{Difficulty-Based Keyless Encryption}\label{section:DBKE}

A difficulty-based key encryption scheme is similar to a symmetric encryption scheme, except that no secret key is kept for use between the encryption and decryption algorithm. 

\begin{definition}[Difficulty-Based Keyless Encryption]\label{def:dbke}
A \emph{difficulty-based keyless encryption (DBKE) scheme} $\Delta$ for a message space $\MsgSp$ with maximum difficulty $\maxdifficulty \in \bbN$ consists of two algorithms:
\begin{itemize}
\item $\Delta.\Enc(\difficulty, \msg) \tor c$: A (probabilistic) encryption algorithm that takes as input difficulty level $\difficulty \le \maxdifficulty$ and message $\msg$ and outputs ciphertext $\ctxt$.
\item $\Delta.\Dec(\ctxt) \to \msg'$: A deterministic decryption algorithm that takes as input ciphertext $\ctxt$ and outputs message $\msg'$ or an error $\bot \not\in\MsgSp$.  	
\end{itemize}
\end{definition}

A DBKE $\Delta$ is \emph{correct} if, for all messages $m \in \MsgSp$ and all difficulty levels $\difficulty \le \maxdifficulty$, we have that $\Pr\left[ \Delta.\Dec(\Delta.\Enc(\difficulty, \msg))=\msg \right] = 1$, where the probability is taken over the randomness of $\Delta.\Enc$.

\begin{figure}[t]
\centering
\scalebox{0.9}{
\fbox{
\begin{tabular}{c|c|c}
\begin{minipage}[t]{0.31\textwidth}
\underline{$\Exp{db-ind}{\Delta,\difficulty}(\Adversary)$:}
\begin{enumerate}
\item $(m_0, m_1, st) \getsr \Adversary(1^\difficulty)$
\item $b \getsr \bit$
\item $c \getsr \Delta.\Enc(\difficulty, m_b)$
\item $b' \getsr \Adversary(c, st)$
\item return $(b' = b)$
\end{enumerate}
\end{minipage}
&~
\begin{minipage}[t]{0.3\textwidth}
\underline{$\Exp{ind}{\Pi}(\Adversary)$:}
\begin{enumerate}
\item $(m_0, m_1, st) \getsr \Adversary()$
\item $\key \getsr \KeySp$
\item $b \getsr \bit$
\item $c \getsr \Pi.\Enc(\key, m_b)$
\item $b' \getsr \Adversary(c, st)$
\item return $(b' = b)$
\end{enumerate}
\end{minipage}
&~
\begin{minipage}[t]{0.29\textwidth}
\underline{$\Exp{key-ind}{\Sigma,\difficulty}(\Adversary)$:}
\begin{enumerate}
\item $(k_0, \wrappedkey) \getsr \Sigma.\Wrap()$
\item $k_1 \getsr \KeySp$
\item $b \getsr \bit$
\item $b' \getsr \Adversary(\wrappedkey, k_0, k_1)$
\item return $(b' = b)$
\end{enumerate}
\end{minipage}
\end{tabular}
}
}
\caption{Security experiments for (\textit{left}) indistinguishability of difficulty-based keyless encryption scheme $\Delta$ at difficulty level $\difficulty$; (\textit{centre}) one-time indistinguishability of symmetric encryption scheme $\Pi$; and (\textit{right}) indistinguishability of difficulty-based keyless key wrap scheme $\Sigma$ with keyspace $\KeySp$ and difficulty level $\difficulty$.}
\label{fig:sec:db-ind}
\label{fig:sec:ind}
\label{fig:sec:key-ind}
\end{figure}

The desired security property for a DBKE is semantic security in the form of ciphertext indistinguishability.  Since there is no persistent secret key, there is no need to consider security notions incorporating chosen plaintext or chosen ciphertext attacks: each plaintext is protected by independent randomness.  The security experiment $\Exp{db-ind}{\Delta,\difficulty}(\Adversary)$ for an adversary $\Adversary$ trying to break indistinguishability of DBKE scheme $\Delta$ at difficulty level $\difficulty$ is shown in \Cref{fig:sec:db-ind}.  We define the advantage of such an adversary in the security experiment as
\shortlongeqn{\Adv{db-ind}{\Delta,\difficulty}(\Adversary) = \left| 2 \cdot \Pr\left[ \Exp{db-ind}{\Delta,\difficulty}(\Adversary) \Rightarrow \true \right] - 1 \right|}
Useful forms of $\Adv{db-ind}{\Delta,\difficulty}(\Adversary)$ will relate the amount of work done by the adversary, the difficulty level, and the adversary's success probability.  

\subsection{Generic construction of DBKE}\label{section:DBKEGenericConstruction}

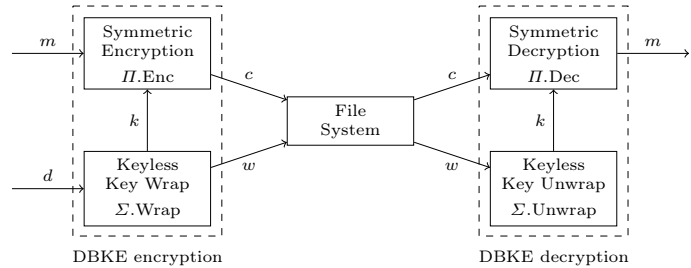
\begin{figure}[b]
\begin{center}
	\scalebox{0.9}{
		\begin{tikzpicture}[yscale=-1,every node/.style={align=center,font=\scriptsize}]
		\node [rectangle,draw,text width=0.65in] (ENC-SYMENC) at (1,1) {Symmetric \\ Encryption \\ \smallskip $\Pi.\Enc$};
		\node [rectangle,draw,text width=0.65in] (ENC-KKW) at (1,3) {Keyless \\ Key Wrap \\ \smallskip $\Sigma.\Wrap$};
		\draw[->] (ENC-KKW) -- node[left] {$\key$} (ENC-SYMENC);
		\draw[->] (-1,1) -- node[above] {$\msg$} (ENC-SYMENC);
		\draw[->] (-1,3) -- node[above] {$\difficulty$} (ENC-KKW);
		\draw[dashed] (-0.1,0.3) -- (2.1,0.3) -- (2.1,3.7) -- (-0.1,3.7) -- (-0.1,0.3);
		\node [below] at (1,3.8) {DBKE encryption};
		\node [rectangle,draw,text width=0.65in] (FS) at (4,2) {File \\ System};
		\draw[->] (ENC-SYMENC) -- node[above] {$\ctxt$} (FS);
		\draw[->] (ENC-KKW) -- node[below] {$\wrappedkey$} (FS);
		\node [rectangle,draw,text width=0.65in] (DEC-SYMENC) at (7,1) {Symmetric \\ Decryption \\ \smallskip $\Pi.\Dec$};
		\node [rectangle,draw,text width=0.65in] (DEC-KKW) at (7,3) {Keyless \\ Key Unwrap \\ \smallskip $\Sigma.\Unwrap$};
		\draw[->] (FS) -- node[above] {$\ctxt$} (DEC-SYMENC);
		\draw[->] (FS) -- node[below] {$\wrappedkey$} (DEC-KKW);
		\draw[->] (DEC-KKW) -- node[left] {$\key$} (DEC-SYMENC);
		\draw[->] (DEC-SYMENC) -- node[above] {$\msg$} (9,1);
		\draw[dashed] (5.9,0.3) -- (8.1,0.3) -- (8.1,3.7) -- (5.9,3.7) -- (5.9,0.3);
		\node [below] at (7,3.8) {DBKE decryption};
		\end{tikzpicture}
	}
\end{center}
\vspace{-1em}
\caption{Architectural diagram for generic construction of a difficulty-based keyless encryption scheme $\Gamma=\Gamma[\Pi,\Sigma]$ from a difficulty-based keyless key wrap scheme $\Sigma$ and a symmetric encryption scheme $\Pi$.}
\label{fig:DBKE}
\vspace{-1em}
\end{figure}

Our main construction of DBKE, as shown in \Cref{fig:DBKE}, generically combines a traditional symmetric encryption scheme with a ``keyless key wrap'', which is difficulty-based form of key wrapping: there is no ``master key'' wrapping the session key, instead the session key is recovered via some difficulty-based operation. In this subsection we present the generic building blocks we use to construct DBKE. In \Cref{section:KKWInstantiation} we show how to instantiate the keyless key wrap.

\begin{definition}[Symmetric encryption scheme]\label{def:symenc}
A \emph{symmetric encryption scheme} $\Pi$ with secret key space $\KeySp = \bit^\sklen$ and message space $\MsgSp$ consists of two algorithms:
\begin{itemize}
\item $\Pi.\Enc(\key, \msg) \tor \ctxt$: A (probabilistic) encryption algorithm that takes as input key $\key \in \KeySp$ and message $\msg \in \MsgSp$ and outputs ciphertext $\ctxt$.
\item $\Pi.\Dec(\key, \ctxt) \to \msg'$: A deterministic decryption algorithm that takes as input key $\key \in \KeySp$ and ciphertext $\ctxt$ and outputs message $\msg' \in \MsgSp$ or an error $\bot \not\in\MsgSp$.
\end{itemize}
\end{definition}

Correctness is defined in the obvious way.  For our purposes, a sufficient security property will be one-time semantic security, in the form of ciphertext indistinguishability.  As above, we will not need to consider security notions incorporating chosen plaintext or chosen ciphertext attacks, since our system will use a key only once.  The security experiment $\Exp{ind}{\Pi}(\Adversary)$ for an adversary $\Adversary$ trying to break indistinguishability of symmetric encryption scheme $\Pi$ is shown in \Cref{fig:sec:ind}.  We define the advantage of such an adversary in the security experiment as
\shortlongeqn{\Adv{ind}{\Pi}(\Adversary) = \left| 2 \cdot \Pr\left[ \Exp{ind}{\Pi}(\Adversary) \Rightarrow \true \right] - 1 \right|}

The second building block for our construction is a keyless key wrap scheme.

\begin{definition}[Keyless key wrap scheme]\label{def:kkw}
A \emph{keyless key wrap scheme} $\Sigma$ for a key space $\KeySp = \bit^\sklen$ with maximum difficulty level $\maxdifficulty \in \bbN$ consists of two algorithms:
\begin{itemize}
\item $\Sigma.\Wrap(\difficulty) \tor (\key, \wrappedkey)$: A (probabilistic) key wrapping algorithm that takes as input difficulty level $\difficulty \le \maxdifficulty$ and outputs key $\key \in \KeySp$ and wrapped key $\wrappedkey$.
\item $\Sigma.\Unwrap(\wrappedkey) \to \key'$: A deterministic key unwrapping algorithm that takes as input wrapped key $\wrappedkey$ and outputs key $\key \in \KeySp$ or an error $\bot\not\in\KeySp$.
\end{itemize}
\end{definition}

Correctness, again, is defined in the natural way: applying $\Unwrap$ to a wrapped key $\wrappedkey$ output by $\Wrap$ should yield, with certainty, the same key $\key$ as originally output by $\Wrap$.

The desirable security property for a keyless key wrap scheme will be indistinguishability of keys: given the wrapped key, can the adversary learn anything about the key within it?  The key indistinguishability security experiment $\Exp{key-ind}{\Sigma,\difficulty}$ for an adversary $\Adversary$ trying to break key indistinguishability of a keyless key wrap scheme at difficulty level $\difficulty$ is shown in \Cref{fig:sec:key-ind}.  We define the advantage of such an adversary in the security experiment as
\shortlongeqn{\Adv{key-ind}{\Sigma,\difficulty}(\Adversary) = \left| 2 \cdot \Pr\left[ \Exp{key-ind}{\Sigma,\difficulty}(\Adversary) \Rightarrow \true \right] - 1 \right|}
As with DBKE security, useful forms of $\Adv{key-ind}{\Sigma,\difficulty}(\Adversary)$ will relate the amount of work done by the adversary, the difficulty level, and the adversary's success probability.

As noted above, we generically construct a difficulty-based keyless encryption scheme by combining a traditional symmetric encryption scheme with a keyless key wrap scheme, as outlined in \Cref{fig:DBKE}.  
Let $\Pi$ be a symmetric encryption scheme with key space $\KeySp = \bit^\sklen$, and let $\Sigma$ be a keyless key wrap scheme for key space $\KeySp$ with maximum difficulty level $\maxdifficulty$.  Construct the difficulty-based keyless encryption scheme $\Gamma[\Pi,\Sigma]$ from $\Pi$ and $\Sigma$ as outlined in \Cref{fig:DBKE} and specified in \Cref{fig:DBKEConstruction}.

\begin{figure}[t]
\centering
\fbox{
\scalebox{0.9}{
\begin{minipage}[t]{0.4\textwidth}
\underline{$\Gamma.\Enc(\difficulty,\msg)$:}
\begin{enumerate}
\item $(\key, \wrappedkey) \getsr \Sigma.\Wrap(\difficulty)$
\item $\ctxt \getsr \Pi.\Enc(\key, \msg)$
\item return $(\ctxt, \wrappedkey)$
\end{enumerate}
\end{minipage}
~
\begin{minipage}[t]{0.4\textwidth}
\underline{$\Gamma.\Dec((\ctxt, \wrappedkey))$:}
\begin{enumerate}
\item $\key' \gets \Sigma.\Unwrap(\wrappedkey)$
\item $\msg' \gets \Pi.\Dec(\key', \ctxt')$
\item return $\msg'$
\end{enumerate}
\end{minipage}
}
}
\caption{Generic construction of a difficulty-based keyless encryption scheme $\Gamma=\Gamma[\Pi,\Sigma]$ from a difficulty-based keyless key wrap scheme $\Sigma$ and a symmetric encryption scheme $\Pi$.}
\label{fig:DBKEConstruction}
\vspace{-1em}
\end{figure}

Our DBKE scheme $\Gamma$ is secure, in the sense of \Cref{fig:sec:db-ind}, under the assumption that the building blocks are secure.  The proof follows from a straightforward game-hopping argument; 
\iffullversion
details are presented in \Cref{appendix:Proofs:generic-DBKE}.
\fi

\begin{theorem}\label{thm:generic}
If $\Sigma$ is a key-indistinguishable difficulty-based keyless key wrap scheme, and $\Pi$ is a one-time indistinguishable symmetric encryption scheme, then $\Gamma=\Gamma[\Pi,\Sigma]$ is a secure difficulty-based keyless encryption scheme.  More precisely, let $\difficulty \le \maxdifficulty$ and let $\Adversary$ be a probabilistic algorithm. Then there exists algorithms $\Bdversary_1$ and $\Bdversary_2$, such that
\shortlongeqn{\Adv{db-ind}{\Gamma,\difficulty}(\Adversary) \le 2 \cdot \Adv{key-ind}{\Sigma,\difficulty}(\Bdversary_1^\Adversary) + \Adv{ind}{\Pi}(\Bdversary_2^\Adversary)}
Moreover, $\Bdversary_1^\Adversary$ and $\Bdversary_2^\Adversary$ have about the same runtime as $\Adversary$.
\end{theorem}

\subsection{Hash-based construction of difficulty-based keyless key wrap}\label{section:KKWInstantiation}

We now show how to construct our difficulty-based keyless key wrap using a hash-based puzzle.  The idea is simple: a random seed $\seed$ is chosen, and the key and a checksum of the seed are derived from the seed using hash functions.  The wrapped key consists of the checksum of the seed and the seed with \emph{some of its bits removed}; the number of bits removed corresponds to the difficulty of the puzzle. This is similar to the sub-puzzle construction of Juels and Brainard \cite{juels1999client} or partial inversion proof of work by Jakobsson and Juels \cite{JJ99}. Such a puzzle is solved by trying all possibilities for the missing bits, in any order and with or without using parallelization.

In particular, let $\sklen \in \bbN$, and let $H_1, H_2 : \bit^\sklen \to \bit^\sklen$ be independent hash functions. Define keyless key wrap scheme $P=P[H_1,H_2]$ as in \Cref{fig:kkwconstruction} (left).  The notation $\seed[\sklen-\difficulty:\sklen]$ on line 2 of $P.\Wrap$ denotes taking the substring of $\seed$ corresponding to indices $\sklen-\difficulty$ up to $\sklen$, removing the first $\difficulty$ bits of $\seed$.

\begin{figure}[b]
\centering
\fbox{
\scalebox{0.9}{
\begin{tabular}{cc|c}
\begin{minipage}[t]{0.25\textwidth}
\underline{$P.\Wrap(d)$:}
\begin{enumerate}
\item $\seed \getsr \bit^\sklen$
\item $\partialseed \gets \seed[\sklen - \difficulty : \sklen]$
\item $\checksum \gets H_1(\seed)$
\item $\key \gets H_2(\seed)$
\item $\wrappedkey \gets (\checksum, \partialseed)$
\item return $(\key, \wrappedkey)$
\end{enumerate}
\end{minipage}
&
\begin{minipage}[t]{0.3\textwidth}
\underline{$P.\Unwrap(\wrappedkey=(\checksum,\partialseed))$:}
\begin{enumerate}
\item $\difficulty \gets \sklen - |\partialseed|$
\item for $i \in \{0,1\}^\difficulty$:
\item \quad\quad $\seed' \gets i \| \partialseed$
\item \quad\quad $\checksum' \gets H_1(\seed')$
\item \quad\quad if $\checksum' = \checksum$:
\item \quad\quad\quad\quad $\key \leftarrow H_2(\seed')$
\item \quad\quad\quad\quad return $\key$
\item return $\bot$
\end{enumerate}
\end{minipage}
&~
\begin{minipage}[t]{0.35\textwidth}
\underline{$\Gamma[\Pi,P].\Degrade(\hat{\ctxt}, \difficulty')$:}
\begin{enumerate}
\item parse $\hat{\ctxt}$ as $(\ctxt, \wrappedkey = (\checksum, \partialseed))$
\item $\difficulty \gets \sklen - |\partialseed|$
\item abort if $\difficulty' < \difficulty$
\item $\partialseed' \gets \partialseed[\difficulty' - \difficulty : |\partialseed|]$
\item $\wrappedkey' \gets (\checksum, \partialseed')$
\item return $(\ctxt, \wrappedkey')$
\end{enumerate}
\end{minipage}
\end{tabular}
}
}
\caption{\textit{Left:} Construction of a hash-based keyless key wrap scheme $P=P[H_1,H_2]$ from hash functions $H_1,H_2$.  \textit{Right:} Degradation algorithm for DBKE $\Gamma=\Gamma[\Pi,P]$ constructed using generic construction $\Gamma$ of \Cref{fig:DBKEConstruction} using hash-based keyless key wrap scheme $P$ of left.}
\label{fig:kkwconstruction}
\label{fig:degrade}
\vspace{-1em}
\end{figure}

The following theorem shows the key indistinguishability security of our hash-based keyless key wrap scheme $P$ in the random oracle model.  The proof consists of a query counting argument in the random oracle model; 
\iffullversion
details are presented in \Cref{appendix:Proofs:kkw}.
\fi

\begin{theorem}\label{thm:kkw}
Let $H_1$ and $H_2$ be random oracles.  Let $\sklen \in \bbN$ and let $\difficulty \le \sklen$.  Let $P=P[H_1,H_2]$ be the keyless key wrap scheme from \Cref{fig:kkwconstruction} (left).  Let $\Adversary$ be an adversary in key indistinguishability experiment against $P$ which makes $q_1$ and $q_2$ distinct queries to its $H_1$ and $H_2$ random oracles, respectively.  Then
\shortlongeqn{\Adv{key-ind}{P,\difficulty}(\Adversary) \le \frac{q_1}{2^{\difficulty-1}} + \frac{2}{2^\difficulty - q_1}}
\end{theorem}

\paragraph{Puzzle granularity.}
The partial pre-image puzzle construction used in \Cref{fig:kkwconstruction} does not allow for fine-grained control of difficulty: removing each additional bit increases the expected computational cost by a factor of 2.  Higher granularity can be achieved similar to how the puzzle difficulty in Bitcoin is set, by giving a hint that narrows the range of data from $2^d$ to some smaller subset.

\subsection{Puzzle Degradation}\label{section:PuzzleDegradation}

We now introduce an additional feature of difficulty-based keyless encryption that emerges naturally from our hash-based keyless key wrap construction: \emph{puzzle degradation}.  Abstractly, puzzle degradation is a process that takes a DBKE ciphertext and increases the difficulty of decrypting it, preferably without needing to decrypt and then re-encrypt at a higher difficulty level.  

In the context of the ArchiveSafe long-term archiving system, this may be used to gradually increase the difficulty of files that have not been accessed for a certain period of time.  For example, a monthly maintenance process could apply degradation to stored files to gradually increase the cost (to both an attacker and an honest party) of accessing increasingly older files.

The DBKE system $\Delta$ from \Cref{def:dbke} is augmented with the algorithm:
\begin{itemize}
\item $\Delta.\Degrade(\ctxt, \difficulty') \tor \ctxt'$: A (possibly probabilistic) algorithm that takes as input ciphertext $\ctxt$ and target difficulty level $\difficulty' \le \maxdifficulty$, and outputs updated ciphertext $\ctxt'$.
\end{itemize}

Correctness is extended to demand that a ciphertext output by $\Delta.\Enc$ then degraded any number of times is still correctly decrypted by $\Delta.\Dec$ (although decryption may take longer).  

Security with the degraded algorithm included should mean, intuitively, that a ciphertext degraded any number of times can be decrypted only using the required amount of work at the new difficulty level.

We capture both correctness and security of degradation formally by demanding that, for all $\difficulty \le \difficulty' \le \maxdifficulty$ and all $\msg \in \MsgSp$, we have that 
\shortlongeqn[;]{\Delta.\Enc(\difficulty', \msg) \equiv \Delta.\Degrade(\difficulty', \Delta.\Enc(\difficulty, \msg))}
in other words: the distribution of ciphertexts produced by encrypting at difficulty $\difficulty'$ is identical to the distribution of ciphertexts produced by encrypting at difficulty $\difficulty$ and then degrading to difficulty $\difficulty'$.

We can achieve degradation in DBKE $\Gamma=\Gamma[\Pi,P]$ constructed from our hash-based keyless key wrap $P$ in a trivial way: by removing $(\difficulty'-\difficulty)$ more bits from the puzzle hint $\partialseed$.  This clearly requires no decryption and re-encryption, only a constant-time edit to the metadata stored containing the wrapped key.  The procedure $\Gamma.\Degrade$ is stated in \Cref{fig:degrade} (\textit{right}).  Degraded ciphertexts are identically distributed to ciphertexts freshly generated at the target difficulty level, as removing additional bits of the partial seed $\partialseed$ is associative. An adversary who possess a copy of the metadata from an earlier version of the archive prior to degradation can solve puzzles and decrypt at the earlier, non-degraded difficulty level.

\subsection{Additional Considerations}
\label{section:OutsourceSolving}
\label{section:Combining}

\lightparagraph{Outsourcing Puzzle Solving.}
The generic DBKE construction $\Gamma$ of \Cref{fig:DBKEConstruction} allows the key unwrapping and ciphertext decryption to be done separately, so the expensive key unwrapping could be outsourced to a cloud server.  In the example of the hash-based keyless key wrap scheme $P$ of \Cref{fig:kkwconstruction}, the user could give the wrapped key $\wrappedkey=(\checksum,\partialseed)$ to the cloud server who unwraps and returns the key $\key$, which the user then locally uses to decrypt the ciphertext $\ctxt$.  

This does mean that the cloud server learns the encryption key $\key$.  
However, this can be avoided with the following adaption to the construction $P$ of \Cref{fig:kkwconstruction}.
During wrapping, the algorithm generates an additional \emph{salt} value $\salt \getsr \{0,1\}^\sklen$ and computes $\key \gets H_2(\seed \| \salt)$; $\salt$ is stored in the wrapped key $\wrappedkey$.  When outsourcing the unwrapping to the cloud server, the user only sends $\checksum$ and $\partialseed$, but not $\salt$.  The cloud server is still able to use the checksum $\checksum$ with the partial seed $\partialseed$ to recover the full seed $\seed$, but lacks the salt $\salt$ and thus the cloud server alone cannot compute the decryption key $\key$.  \Cref{thm:kkw} still applies to this adaptation.

\lightparagraph{Combining Keyless and Keyed Encryption.}
As previously mentioned, our keyless encryption approach can (and should) be used in conjunction with traditional keyed encryption mechanisms using a different set of keys. Traditional keyed encryption gives honest parties a (conjecturally exponential) work factor advantage over adversaries if keys remain uncompromised, while keyless encryption slows adversaries if the traditional encryption keys are compromised. The two schemes can be layered in one of two ways: first applying keyless encryption $\DBKE$ and encrypting the result using keyed symmetric  encryption $\Sym$ (i.e., $c \gets \Sym.\Enc(k, \DBKE.\Enc(d, m))$) or in the order order, with keyless encryption on the outer layer (i.e., $c \gets \DBKE.\Enc(d, \Sym.\Enc(k, m))$).  Either approach yields robust confidentiality, but we recommend the latter method as it facilitates the puzzle degradation process described in \Cref{section:PuzzleDegradation}.

\section{Evaluation}\label{section:Evaluation}
We evaluate ArchiveSafe by measuring its performance against other systems through real life experiment. The goals of the experiment are to: (1) measure the overhead ArchiveSafe introduces on adversaries and honest users, and (2) verify that puzzle solving difficulty scale according to the theoretical system design.

\subsection{Prototype Implementation}
To run the evaluation experiment, we implemented a prototype of ArchiveSafe.  In terms of instantiating the difficulty-based keyless encryption using the generic construction from \Cref{section:DBKEGenericConstruction}, our proof-of-concept uses AES-128 in CBC mode for the symmetric encryption scheme.  The hash functions $H_1$ and $H_2$ in the hash-based keyless key wrap scheme are both instantiated with Argon2id \cite{argon2} with a prefix byte acting as a domain separator between $H_1$ and $H_2$, with the following parameters: parallelism level: 8; memory: 102,400\,KiB; iterations: 2; output length: 128 bits. We did not parallelize puzzle solving in $\Unwrap$ to avoid locking other system operations, but it is easily parallelized.

The ArchiveSafe prototype is implemented as a Linux Filesystem in Userspace (FUSE) using a Python toolkit\footnote{\url{https://github.com/skorokithakis/python-fuse-sample}} to simplify implementation.  Our Python FUSE driver relies on the OpenSSL library for encryption and decryption, and Ubuntu's \texttt{argon2} package.  In a real deployment in the context of a filesystem, ArchiveSafe would be implemented as a kernel module, likely written in C, for improved performance and reliability.

Our prototype has a tuneable difficulty level, which we label in this section as D1, D2, D3, etc.  Difficulty D$x$ corresponds to hash-based keyless key wrap scheme $P$ of \Cref{fig:kkwconstruction} with difficulty parameter $d=4x$; in other words, D1 removes 4 bits of the seed, D2 removes 8 bits of the seed, etc. We chose a 4-bit step between difficulty levels to focus on how system behaviour scales across difficulty levels; finer gradations could be chosen by users.

\subsection{Experimental Setup}
The experiment measures ArchiveSafe's performance at three difficulty levels (D1, D2, D3) compared to an unencrypted file system (denoted UN) and Linux's built-in folder encryption using eCryptfs\footnote{https://www.ecryptfs.org/} (denoted FE) and disk encryption (denoted DE) on read and write tasks at different file sizes.  When running the ArchiveSafe experiments, the ArchiveSafe FUSE driver was writing its files to an unencrypted file system.

\lightparagraph{Measurements.}
For each storage system being evaluated, we measure \emph{read} and \emph{write} times for files of sizes 1\,KB, 100\,KB, 1\,MB, 10\,MB, and 100\,MB. Performance is measured at the application level, from the time the file is opened until the time the read/write operation is completed.  For folder and disk encryption, this includes the filesystem's encryption operations.  For ArchiveSafe, we instrumented the driver to record the total time as well times for different sub-tasks (encryption, puzzle solving, decryption, file system I/O).

\lightparagraph{Test environment.}
Measurements were performed on a single-user Linux machine with no other processes running. The computer was a MacBook Pro running Ubuntu Linux 18.04 LTS with an 4-core Intel Core i7-4770HQ processor with base frequency 2.2\,GHz, bursting to 3.4\,GHz.  The computer had 16\,GiB of RAM.  The hard drive was a 256\,GiB solid state drive with 512-byte logical sectors and 4096-byte physical sectors.  The disk encryption was done using Linux Unified Key Setup system version 2.0, and folder encryption was done using the Enterprise Cryptographic Filesystem (eCryptfs) version 5.3.

\lightparagraph{Execution.}
For each storage system and file size, we performed many repetitions of the following tasks.  A file was created with randomly generated alphanumeric characters using a non-cryptographic random number generator. Read and write operations were measured as indicated above.
For file sizes of 1\,KB, 100\,KB, 1\,MB, and 10\,MB, we collected data for 1000 writes and reads; for 100\,MB files, we ran 200 writes and reads, due to extensive time of operations at this size.

\subsection{Results}

\begin{table}[t]
\centering
\scalebox{0.8}{
\begin{tabular}{lrrrrrr|rrrrrr}
  \toprule
  \multirow{2} {*} {\textbf{File system}} &
  \multicolumn {6}{c|}{\textbf{Read}} &   \multicolumn {5}{c}{\textbf{Write}} \\
    & 1\,KB & 100\,KB & 1\,MB & 10\,MB & 100\,MB & & 1\,KB & 100\,KB & 1\,MB & 10\,MB & 100\,MB \\
  \midrule
    Unencrypted (UN) & 0.526 & 0.550 & 1.70 & 10.1 & 110 & & 0.07 & 0.25 & 0.85 & 6.76 & 97.82\\
    Disk Encryption (DE) & 0.737 & 0.924 & 3.15 & 10.5 & 160 & & 0.08 & 0.25 & 0.83 & 6.63 & 97.97\\
    Folder Encryption (FE) & 0.737 & 0.961 & 3.42 & 10.9 & 190 & & 0.12 & 0.50 & 3.31 & 29.07 & 319.88\\
\midrule
    ArchiveSafe D1 & 630 & 630 & 630 & 650 & 860 & & 141.05 & 141.67 & 146.09 & 221.73 & 848.30\\
    ArchiveSafe D2 & 7070 & 7080 & 7310 & 7180 & 7290 & & 141.25 & 141.43 & 145.08 & 223.50 & 847.02\\
	ArchiveSafe D3 & ~112140 & ~111760 & ~107390 & ~114530 & ~107630 & & ~141.01 & ~140.98 & ~145.74 & ~222.40 & ~846.06\\
  \bottomrule
\end{tabular}
}
\smallskip
\caption{Average read and write times in milliseconds}
\label{tab:UNDED1D2D3readavg}
\vspace{-1em}
\end{table}

\begin{table}[t]
\begin{minipage}{0.60\textwidth}
\centering
\scalebox{0.85}{
\begin{tabular}{llrrrrr}
  \toprule
    Diff.& & 1\,KB & 100\,KB & 1\,MB & 10\,MB & 100\,MB \\
  \midrule
     D1 & Puzzle Solve & 510 & 510 & 510 & 510 & 500\\
	     & Decryption & 5.42 & 5.71 & 7.25 & 20 & 150\\
	     & Other & 0.387 & 0.373 & 0.378 & 0.384 & 0.363\\
	     \midrule
     D2 & Puzzle Solve & 6960 & 6980 & 7210 & 7050 & 6930\\
	     & Decryption & 5.58 & 6.12 & 7.89 & 20 & 140\\
	     & Other & 0.357 & 0.373 & 0.376 & 0.374 & 0.335\\
	     \midrule
     D3 & Puzzle Solve & ~112040 & ~111730 & ~107280 & ~114410 & ~107270\\
	     & Decryption & 5.56 & 5.94 & 7.96 & 20 & 140\\
	     & Other & 1.075 & 1.216 & 0.971 & 1.195 & 1.045\\
  \bottomrule
\end{tabular}
}
\smallskip
\caption{Read sub-tasks average times in milliseconds}
\label{tab:D1D2D3detavg}
\end{minipage}
~
\begin{minipage}{0.38\textwidth}
\centering
\scalebox{0.85}{
\begin{tikzpicture}
    \begin{axis}[ylabel={Time in ms},width=0.99\textwidth,symbolic x coords={0, D1,D2,D3},
        xmin={[normalized]0},
        xmax={[normalized]4},
    xtick=data,ymin=0, ymax=400000, ymode=log,
    height=2.13in
    ]
        \addplot [ybar=7pt, bar width=8pt, error bars/.cd, y dir=both, y explicit]
            table [col sep=comma, x=File, y=Mean, y error=STDev]{puzzle.csv};
    \end{axis}
\end{tikzpicture}
}

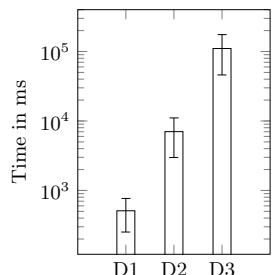
\captionof{figure}{Puzzle solving time in milliseconds (average, standard deviation)}
\label{fig:Puzzleavg}
\end{minipage}
\vspace{-2em}
\end{table} 

\Cref{tab:UNDED1D2D3readavg} shows average read and write times for the file systems under consideration at different file sizes.  
Since read operations in the ArchiveSafe system become increasingly expensive with difficulty, we show in \Cref{tab:D1D2D3detavg} the average time of sub-tasks of ArchiveSafe read operations at different file sizes and difficulties: the puzzle solving time (which should scale with puzzle difficulty), the system file read time plus decryption time (which should scale with file size), and the overhead from other file system driver operations (which includes puzzle read and system file open times).  As the partial pre-image puzzle used in ArchiveSafe leads to highly variable solving times, \Cref{fig:Puzzleavg} shows the average time and standard deviation for puzzle solving at difficulties D1, D2, and D3.

\subsection{Discussion}\label{section:Discussion}
The results show consistent behaviour across different file sizes. The larger files consumed more time in decrypting and reading. We also observed that the time consumed is roughly the same for smaller file sizes (1\,KB and 100\,KB) where operation cost is dominated by overhead.

As expected, the read speeds decrease with the difficulty level because the system must solve the puzzle before reading the file and the puzzle solving effort scales with the difficulty level.  As per \Cref{tab:D1D2D3detavg}, puzzle solve times on average scale by a factor of 13.6--14.1$\times$ between D1 and D2 and a factor of 14.9--16.2$\times$ between D2 and D3, roughly in line with the theoretical scaling factor of 16$\times$.

Evaluating the overhead added by ArchiveSafe for write operations, we see in \Cref{tab:UNDED1D2D3readavg} that ArchiveSafe incurs a baseline overhead related to setting up the puzzle (which involves 2 Argon2 calls), then scales with the file size due to the cost of AES encryption and writing.   Note that ArchiveSafe uses a different encryption library (user-space calls to OpenSSL) compared with disk and file encryption (kernel encryption via dm-crypt), so symmetric encryption/decryption performance is not directly comparable, but we see similar scaling.

The short summary of performance is that ArchiveSafe adds a 140--520\,ms overhead when writing a file, and a customizable overhead when reading a file, ranging from 510\,ms at difficulty D1, 7 seconds at D2, or 110 seconds at D3.  But recall that adding computational overhead at read time is exactly the purpose of ArchiveSafe!  What an acceptable difficulty level---and hence acceptable computational overhead at read time for honest users---is a policy choice by the system administrator. As noted earlier, choosing the difficulty level depends on the tolerable cost for honest users to access data, the perceived risk of a data breach, and the anticipated value of the information to an adversary, and is a calculation that must be left to the adopter.  Note that honest users need not solely rely on sequential operations on their own computer: as described in \Cref{section:OutsourceSolving} an ArchiveSafe installation could be configured so that honest users offload their puzzle solving tasks to private or commercial clouds which are spun up on demand with large amounts of parallelization to reduce the wall clock time before they can access a file.

\begin{table}[t]
\centering
\scalebox{0.95}{
\begin{tabular}{lrrrr}
  \toprule
    &\multicolumn{1}{c}{D3} &\multicolumn{1}{c}{D4} & \multicolumn{1}{c}{D5} & \multicolumn{1}{c}{D6} \\
  \midrule
  \multicolumn{5}{l}{\textit{Honest user decrypting 1 file}} \\
  Local machine, threaded 4 cores, 2.2\,GHz & ~~~0.5\,min. & 7.3\,min. & 2\,hrs. & 31\,hrs. \\
  Cloud server \texttt{c5.metal}, spot pricing & $\ll$\$0.01 & $<$\$0.01 & \$0.05 & \$0.73 \\
  \midrule
  \multicolumn{5}{l}{\textit{Adversary decrypting 1 million files}} \\
  Cloud server \texttt{c5.metal} & 8\,days & ~~130\,days & 5.7\,yrs. & 91.4\,yrs. \\
  Cloud server \texttt{c5.metal}, spot pricing & \$178 & \$2,852 & ~~\$45,648 & ~~\$730,364 \\
  \bottomrule
\end{tabular}
}
\smallskip
\caption{Dollar cost and computation time required to unlock ArchiveSafe files}
\label{tab:cost}
\vspace{-2em}
\end{table}

\Cref{tab:cost} shows examples of costs at higher difficulty levels.  To provide further interpretation to these costs, we look not only at the computation time required for an honest user on our test platform to decrypt a file, but also at the real-world cost for an adversary, based on the cost of renting computation time on Amazon Web Services (AWS) Elastic Cloud Compute (EC2) platform.  EC2 has many machine types available; Argon2 is designed to not be substantially accelerated by more sophisticated architectures, GPUs, or ASICs.  As such we choose for our pricing example an EC2 instance that minimizes cost per core-GHz-hour; the \texttt{c5.metal} EC2 instance type has 96 Intel Xeon cores running at 3.6\,GHz at a cost of USD\$0.9122 per hour using Amazon's cheapest spot pricing model.\footnote{\url{https://aws.amazon.com/ec2/instance-types/}, \\\url{https://aws.amazon.com/ec2/spot/pricing/}; prices as of April 23, 2020.}

We can see, for example, that at difficulty D5, an honest user can unlock an archived file with about 2 hours of work on a local machine, or about 3 minutes of \texttt{c5.metal} rental costing 4.5 cents at spot pricing (20 cents on-demand pricing).  However, an adversary trying to decrypt 1 million such files from a data breach would need 5.7 years of \texttt{c5.metal} rental at a spot pricing cost of USD\$45,648.

\section{Conclusion}\label{section:Conclusion}

ArchiveSafe, using difficulty-based keyless encryption, can add defense-in-depth to confidentiality of archived data and change the economics of mass leakage attacks via data breaches.
We expect that most uses of ArchiveSafe would be in addition to, not as a replacement for, traditional keyed encryption; full cryptographic security would be achieved if encryption keys are properly managed and kept safe, but ArchiveSafe provides a residual level of protection if traditional encryption keys are also breached.
This means the key management service is no longer a single point of failure.

One target application is IT systems which retain large amounts of archival data, most of which will be rarely or perhaps never again accessed by legitimate users. 
Although honest users have no advantage in difficulty-based decryption compared to an adversary on a file-by-file basis, if their operational goals are different---an honest user decrypting 1 file occasionally, versus an adversary decrypting thousands or millions of files quickly---their costs are different.

Our approach can be applied in a variety of system architectures:
local storage and execution (as demonstrated by our prototype), local storage with private or public cloud assistance for puzzle solving, or remote (file server / cloud) storage with local or assisted puzzle solving.
Our approach can also apply to different storage paradigms, including file systems, cloud ``blob'' storage, and databases.

Puzzle difficulty can be set as a system-wide or with higher granularity based individual records' sensitivity.  A novel features of our construction is the ability to degrade puzzle difficulty effectively for free, which could be built into periodic maintenance or through a heuristic system based on suspicious activity.

\subsection*{Acknowledgements}

This work grew out of earlier discussions on use of puzzles for database encryption with Farhad Moghimifar, Suriadi Suriadi, and Ernest Foo at the Queensland University of Technology. R.S. is supported by Natural Sciences and Engineering Research Council of Canada (NSERC) Discovery grant RGPIN-2016-06062.  D.S. is supported by NSERC Discovery grant RGPIN-2016-05146 and NSERC Discovery Accelerator Supplement grant RGPIN-2016-05146.

\iffullversion
\appendix
\section{Proofs}\label{appendix:Proofs}

\subsection{Security of Generic DBKE Construction $\Gamma$}\label{appendix:Proofs:generic-DBKE}

\begin{proof}[of \Cref{thm:generic}]
The security proof proceeds as a sequence of games.  For Game $G_i$, let $S_i$ denote the event that game $G_i$ outputs true.  Ley $\KeySp$ be the key space of the symmetric encryption scheme $\Pi$, which is also the key space of the keyless key wrap scheme $\Sigma$.

\paragraph{Game 0.}
Denoted $G_0$, Game 0 as shown in the left side of \Cref{fig:pf:generic:games} is the db-ind experiment from \Cref{fig:sec:db-ind} with construction $\Gamma=\Gamma[\Pi,\Sigma]$ inline.  Thus,
\begin{equation}\label{eqn:thm:generic:0}
\Pr\left[ \Exp{db-ind}{\Gamma,\difficulty}(\Adversary) \Rightarrow \true \right] = \Pr[S_0] \enspace .
\end{equation}

\paragraph{Game 1.}
In this game, the challenger generates two symmetric encryption keys $k$ and $k'$; it uses $k$ in the key wrapping scheme, but $k'$ in the symmetric encryption scheme. This is shown in Game $G_1$ in the right side of \Cref{fig:pf:generic:games}.

\begin{figure}[t]
\centering
\scalebox{0.95}{
\fbox{
\begin{minipage}[t]{0.45\textwidth}
\underline{Game $G_0$:}
\begin{enumerate}
\item $(m_0, m_1, st) \getsr \Adversary(1^\difficulty)$
\item $b \getsr \bit$
\item $(\key, \wrappedkey) \getsr \Sigma.\Wrap(\difficulty)$
\item $\ctxt \getsr \Pi.\Enc(\key, \msg_b)$
\item $b' \getsr \Adversary((\ctxt, \wrappedkey), st)$
\item return $(b' = b)$
\end{enumerate}
\end{minipage}
~
\begin{minipage}[t]{0.45\textwidth}
\underline{Game $G_1$:}
\begin{enumerate}
\item $(m_0, m_1, st) \getsr \Adversary(1^\difficulty)$
\item $b \getsr \bit$
\item $(\key, \wrappedkey) \getsr \Sigma.\Wrap(\difficulty)$
\item \gamechange{$\key' \getsr \KeySp$}
\item $\ctxt \getsr \Pi.\Enc(\gamechange{$\key'$}, \msg_b)$
\item $b' \getsr \Adversary((\ctxt, \wrappedkey), st)$
\item return $(b' = b)$
\end{enumerate}
\end{minipage}
}
}
\caption{Sequence of games for proof of Theorem~\ref{thm:generic}.  Changes between games are \gamechange{highlighted}.}
\label{fig:pf:generic:games}
\vspace{-1em}
\end{figure}

First we show in \Cref{claim:generic:1} that Game 0 and Game 1 are indistinguishable under the assumption that the key wrapping scheme is secure.  Then we argue in \Cref{claim:generic:2} that breaking Game 1 corresponds to breaking the indistinguishability of the symmetric key encryption scheme.

\begin{claim}\label{claim:generic:1}
Let $\Bdversary_1$ be the algorithm shown in \Cref{fig:pf:generic:B}, which is an adversary against the key indistinguishability of keyless key wrap scheme $\Sigma$. Then
\begin{equation}\label{eqn:claim:generic:1}
\left| \Pr[S_0] - \Pr[S_1] \right| \le \Adv{key-ind}{\Sigma,\difficulty}(\Bdversary_1^\Adversary) \enspace .
\end{equation}
\end{claim}

\begin{proof}
$\Bdversary_1$'s input is a challenge $(\wrappedkey,k_0,k_1)$ from a challenger for the key indistinguishability of difficulty-based keyless key wrap scheme $\Sigma$.  This means that $\wrappedkey$ is the wrapping of either $k_0$ or $k_1$, chosen by a random hidden bit $b$ in $\Exp{key-ind}{\Sigma,\difficulty}$.  When the hidden bit $b=0$, and hence when $w$ is the wrapping of $k_0$, then, in the ciphertext $(\ctxt,\wrappedkey)$ that $\Bdversary_1$ gives to $\Adversary$, the key used in the key wrapping is \emph{the same as} the key used in the symmetric encryption scheme, so $\Bdversary_1$ exactly simulates Game 0 to $\Adversary$.  When the hidden bit $b=1$, and hence when $\wrappedkey$ is the wrapping of $k_1$, then, in the ciphertext $(\ctxt,\wrappedkey)$ that $\Bdversary_1$ gives to $\Adversary$, the key used in the key wrapping is \emph{different from} the key used in the symmetric encryption scheme, so $\Bdversary_1$ exactly simulates Game 1 to $\Adversary$.  Thus, if $\Adversary$ outputs $\hat{b}'$ with different probabilities in Game 0 compared to Game 1, then $\Bdversary_1^\Adversary$ outputs $\hat{b}'$ with different probabilities when the hidden bit $b$ is 0 or 1. This shows that \cref{eqn:claim:generic:1} holds. \squareforqed
\end{proof}

\begin{figure*}[t]
\centering
\scalebox{0.95}{
\fbox{
\begin{minipage}[t]{0.32\textwidth}
\underline{$\Bdversary_1^\Adversary(w,k_0,k_1)$:}
\begin{enumerate}
\item $(m_0, m_1, st) \getsr \Adversary(1^\difficulty)$
\item $\hat{b} \getsr \bit$
\item $\ctxt \getsr \Pi.\Enc(k_0, m_{\hat{b}})$
\item $\hat{b}' \getsr \Adversary((\ctxt, \wrappedkey), st)$
\item if $(\hat{b}' = \hat{b})$ return 0
\item else return 1
\end{enumerate}
\end{minipage}
~
\begin{minipage}[t]{0.32\textwidth}
\underline{$\Bdversary_2^\Adversary()$:}
\begin{enumerate}
\item $(m_0, m_1, st) \getsr \Adversary(1^\difficulty)$
\item return $(m_0, m_1, st)$
\end{enumerate}
\end{minipage}
~
\begin{minipage}[t]{0.3\textwidth}
\underline{$\Bdversary_2^\Adversary(\ctxt, st)$:}
\begin{enumerate}
\item $(\hat{\key}, \wrappedkey) \getsr \Sigma.\Wrap(\difficulty)$
\item $\hat{b}' \getsr \Adversary((\ctxt, \wrappedkey), st)$
\item return $\hat{b}'$
\end{enumerate}
\end{minipage}
}
}
\caption{Reductions for the proof of \Cref{thm:generic}.}
\label{fig:pf:generic:B}
\vspace{-1em}
\end{figure*}

\begin{claim}\label{claim:generic:2}
Let $\Bdversary_2$ be the algorithm shown in \Cref{fig:pf:generic:B}, which is an adversary against the one-time indistinguishability of symmetric encryption scheme $\Pi$. Then
\begin{equation}\label{eqn:claim:generic:2}
\Pr[S_1] \le \Pr\left[\Exp{ind}{\Pi}(\Bdversary_2^\Adversary) \Rightarrow \true\right] \enspace .
\end{equation}
\end{claim}

\begin{proof}
$\Bdversary_2^\Adversary$ is an adversary in the security experiment $\Exp{ind}{\Pi}(\Bdversary_2^\Adversary)$ for the one-time indistinguishability of symmetric encryption scheme $\Pi$.  When we inline the code of $\Bdversary_2^\Adversary$ in $\Exp{ind}{\Pi}(\Bdversary_2^\Adversary)$, we see that it is exactly the same code as Game 1, except some lines are reordered, and some variables are named differently.  In particular, $\Adversary$ is run with a ciphertext $\ctxt$ that is the encryption of either $m_0$ or $m_1$ under the key $k$ from the symmetric encryption experiment, but this key is different from the key $\hat{\key}$ that is in the wrapped key $\wrappedkey$ that $\Adversary$ is provided with.  Thus, \cref{eqn:claim:generic:2} holds. \squareforqed
\end{proof}

Combining equations \eqref{eqn:thm:generic:0}, \eqref{eqn:claim:generic:1}, and \eqref{eqn:claim:generic:2}, we get
\begin{align*}
\Adv{db-ind}{\Gamma,\difficulty}(\Adversary) 
&= \left| 2 \cdot \Pr\left[ \Exp{db-ind}{\Gamma,\difficulty}(\Adversary) \Rightarrow \true\right] - 1 \right| \\
&= \left| 2 \cdot \Pr[S_0] - 1 \right| \tag{by \eqref{eqn:thm:generic:0}} \\
&= \left| 2 \cdot (\Pr[S_0] - \Pr[S_1] + \Pr[S_1]) - 1 \right| \\
&\le 2 \left| \Pr[S_0] - \Pr[S_1] \right| + \left| 2 \cdot \Pr[S_1] - 1 \right| \\
&\le 2 \Adv{key-ind}{\Sigma,\difficulty}(\Bdversary_1^\Adversary) + \left| 2 \cdot \Pr[S_1] - 1 \right| \tag{by \eqref{eqn:claim:generic:1}} \\
&\le 2 \Adv{key-ind}{\Sigma,\difficulty}(\Bdversary_1^\Adversary) + \left| 2 \cdot \Pr[\Exp{ind}{\Pi}(\Bdversary_2^\Adversary) \Rightarrow \true] - 1 \right| \tag{by \eqref{eqn:claim:generic:2}} \\
&= 2 \cdot \Adv{key-ind}{\Sigma,\difficulty}(\Bdversary_1^\Adversary) + \Adv{ind}{\Pi}(\Bdversary_2^\Adversary)
\end{align*}
which is the desired result.

By inspection of $\Bdversary_1^\Adversary$ and $\Bdversary_2^\Adversary$, we can see that their runtimes are the runtime of $\Adversary$ plus a minimal cost of either encryption or wrapping.
\qed
\end{proof}

\subsection{Security of Hash-Based Keyless Key Wrap Scheme $P$}\label{appendix:Proofs:kkw}

\begin{proof}[of \Cref{thm:kkw}]
Let $k_0$, $k_1$, and $\wrappedkey$ be as in $\Exp{key-ind}{P,\difficulty}$ in \Cref{fig:sec:key-ind} for keyless key wrap scheme $P$, so that $\seed$ is the seed behind $k_0$ and $\wrappedkey$. 

Let $W$ be the event that $\Exp{key-ind}{P,\difficulty}(\Adversary)$ outputs $\true$.  Let $E_i$ be the event that $\Adversary$ queries $\seed$ to random oracle $H_i$, for $i=1,2$.  Our task is to bound $\Pr[W]$, which we do using the following application of the law of total probability:
\begin{align*}
\Pr[W] 
=& \Pr[W|\lnot E_2] \cdot \Pr[\lnot E_2] \\ 
 & + \Pr[W|E_2 \land E_1] \cdot \Pr[E_2 \land E_1] + \Pr[W|E_2 \land \lnot E_1] \cdot \Pr[E_2 \land \lnot E_1]
\end{align*}

If $E_2$ does not occur, then, since $k_0=H_2(\seed)$, $\Adversary$ has no information about $k_0$ and thus has no advantage in distinguishing $k_0$ from $k_1$, so $\Pr[W|\lnot E_2] = \frac{1}{2}$.

Next, we observe that $\Pr[E_2 \land E_1] \le \Pr[E_1]$.  The only way $\Adversary$ can learn information about $\seed$ (and hence $\key$) is by querying values to $H_1$, and since $H_1$ is a random oracle, the adversary can rule out at most one guess for the missing $\difficulty$ bits of $\seed$ with each query to $H_1$.  Thus $\Pr[E_1] \le \frac{q_1}{2^\difficulty}$.

Now we observe that $\Pr[E_2 \land \lnot E_1] = \Pr[E_2 | \lnot E_1] \Pr[\lnot E_1] \le \Pr[E_2 | \lnot E_1]$.  Since the $q_1$ queries to $H_1$ could have ruled out $q_1$ candidate values for the missing bits of $\seed$, we have that $\Pr[E_2 | \lnot E_1] \le \frac{q_2}{2^\difficulty - q_1}$.  Additionally, we note that, when $E_2 \land \lnot E_1$ occurs, $\Adversary$ has no information to help it determine which of its $q_2$ queries to $H_2$ caused $E_2$ to occur, so $\Pr[W|E_2 \land \lnot E_1] = \frac{1}{q_2}$.

Substituting the above observations into the expression for $\Pr[W]$, and bounding all other probabilities by 1, we get
\[ \Pr[W] \le \left(\frac{1}{2} \cdot 1\right) + \left(1 \cdot \frac{q_1}{2^\difficulty}\right) + \left(\frac{1}{q_2} \cdot \frac{q_2}{2^\difficulty - q_1}\right) = \frac{1}{2} + \frac{q_1}{2^\difficulty} + \frac{1}{2^\difficulty-q_1} \enspace ; \]
substituting into the advantage expression yields the desired result.
\qed
\end{proof}

\fi

\end{document}